\documentclass[runningheads, envcountsame, a4paper]{llncs} 
\usepackage{amssymb}
\usepackage{amsmath}
\usepackage{comment}
\usepackage{algorithm}
\usepackage{algorithmic}
\usepackage{stmaryrd}
\SetSymbolFont{stmry}{bold}{U}{stmry}{m}{n}
\usepackage{paralist}
\usepackage{bm}
\usepackage{multirow}
\usepackage[trim]{tokenizer}

\usepackage{tikz}
\usetikzlibrary{arrows,automata,shapes,shapes.multipart}

\usepackage{xspace}

\newcommand{\defemph}[1]{\textbf{\emph{#1}}}
\newcommand{\setsep}{\; | \;}

\newcommand{\nhalf}{\lfloor \frac{n}{2} \rfloor}
\newcommand{\tuple}[1]{{\langle#1\rangle}}

\newcommand{\gennetwork}{\tuple{L_1, \ldots, L_d}}
\newcommand{\netoutput}{\mathop{outputs}}
\newcommand{\depth}{\mathop{depth}}

\newcommand{\once}{\mathop{once}}
\newcommand{\valid}{\mathop{valid}}
\newcommand{\sorts}{\mathop{sorts}}
\newcommand{\updatechannel}{\mathop{update}}
\newcommand{\used}{\mathop{used}}
\newcommand{\okienka}{\mathop{subnets}}
\newcommand{\okienkaoverlap}{\mathop{overlapsubnets}}
\newcommand{\netoverlap}{\mathop{overlap}}

\newcommand{\bb}{\mathbb{B}}

\newcommand{\zz}{\mathbb{Z}}

\newcommand{\nop}[1]{}

\newcounter{sncolumncounter}
\newcounter{snrowcounter}

\def \nodelabel#1{%
\setcounter{snrowcounter}{1}
 \foreach \i in {#1}{%
   \draw (\value{sncolumncounter},\value{snrowcounter}) node[anchor=south]{\i};
   \addtocounter{snrowcounter}{1}
 }
\addtocounter{snrowcounter}{-1}
 \addtocounter{sncolumncounter}{1}
}

\def \addcomparator#1#2{%
    \draw (\value{sncolumncounter},#1) node[circle,fill=black,minimum size=4pt,inner sep=0pt,outer sep=0pt]{}--(\value{sncolumncounter},#2) node[circle,fill=black,minimum size=4pt,inner sep=0pt,outer sep=0pt]{};
}
 
\def \addlayer{%
  \addtocounter{sncolumncounter}{1}
}

\def \nextlayer{%
  \draw [dashed] (\value{sncolumncounter}+0.5,0.6)--(\value{sncolumncounter}+0.5,\value{snrowcounter}+1);
  \addlayer{}
}

\newenvironment{sortingnetwork}[2]
{
  \setcounter{sncolumncounter}{0}
  \setcounter{snrowcounter}{#1}
  \def \sn@fullsize{15}
  \begin{tikzpicture}[scale=#2*0.7]
}
{
  \foreach \i in {1, ..., \value{snrowcounter}}
  {
    \draw (-0.5,\i)--(\value{sncolumncounter}+0.5,\i);
  }
  \end{tikzpicture}
}
\makeatother

\begin{document}

\title{Optimal Sorting Networks}
\author{
Daniel Bundala \and Jakub Z\'{a}vodn\'{y}\\
\institute{Department of Computer Science, University of Oxford \\ Wolfson Building, Parks Road, Oxford, OX1 3QD, UK}
\email{$\{$daniel.bundala,jakub.zavodny$\}$@cs.ox.ac.uk}
}
\authorrunning{D.~Bundala \and J.~Z\'{a}vodn\'{y}}
\toctitle{Optimal Sorting Networks}
\tocauthor{Daniel~Bundala and Jakub~Z\'{a}vodn\'{y}}
\maketitle
\setcounter{footnote}{0}

\begin{abstract}
This paper settles the optimality of sorting networks given in The Art of Computer Programming vol. 3 more than 40 years ago. The book lists efficient sorting networks with $n \leq 16$ inputs. In this paper we give general combinatorial arguments showing that if a sorting network with a given depth exists then there exists one with a special form. We then construct propositional formulas whose satisfiability is necessary for the existence of such a network. Using a SAT solver we conclude that the listed networks have optimal depth. For $n \leq 10$ inputs where optimality was known previously, our algorithm is four orders of magnitude faster than those in prior work.
\end{abstract}

\section{Introduction}

In their celebrated result, Ajtai, Koml\'{o}s and Szemer\'{e}di (AKS)~\cite{AKS}, gave an optimal oblivious sorting algorithm with $O(n \log n)$ comparisons in $O(\log n)$ parallel steps. An oblivious sorting algorithm is one in which the order of comparisons is fixed and depends only on the number of inputs but not their values. Compare this with standard algorithms such as MergeSort or QuickSort where the order of comparisons crucially depends on the input values.

A popular model of oblivious sorting algorithms are so-called sorting networks, which specify a sequence of swap-comparisons on a set of inputs, and whose depth models the number of parallel steps required. Even though the AKS network has asymptotically optimal depth, it is infamous for the large constant hidden in the big $O$ bound; recursively constructed networks of depth $O(\log^2n)$~\cite{Batcher} prove superior to the AKS network for all practical values of $n$. Small networks for small numbers of inputs serve as base cases for these recursive methods. However, constructing networks of optimal depth has proved extremely difficult (e.g.,~\cite{Germans,Parberry}) and is an open problem even for very small number of inputs. We address this problem in this paper.

Already in the fifties and sixties various constructions appeared for small sorting networks on few inputs. In 1973 in The Art of Computer Programming vol.~3~\cite{Knuth}, Knuth listed the best sorting networks with $n \leq 16$ inputs known at the time. It was further shown in~\cite{BoseNelson} that these networks have optimal depth for $n \leq 8$. No progress had been made on the problem until 1989 when Parberry~\cite{Parberry} showed that the networks listed in \cite{Knuth} are optimal for $n=9$ and $n=10$. The result was obtained by implementing an exhaustive search with pruning based on symmetries in the first two parallel steps in the sorting networks, and executing the algorithm on a supercomputer (Cray-2). Despite the great increase in available computational power in the 24 years since, the algorithm would still not be able to handle the case $n=11$. Recently there were attempts~\cite{Germans} at solving the case $n=11$ but we are not aware of any successful one.

Forty years after the publication of the list of small sorting networks by Knuth~\cite{Knuth}, we finally settle their optimality for the remaining cases $n=11$ up to and including $16$. We give general combinatorial arguments showing that if a small-depth sorting network exists then there exists one with a special form. We then construct propositional formulas whose satisfiability is necessary for the existence of such a network. By checking the satisfiability of the formulas using a SAT solver we conclude that no smaller networks than those listed exist.

We obtained all our results using an off-the-shelf SAT solver running on a standard desktop computer. It is noteworthy that our algorithm required less than a second to prove the optimality of networks with $n \leq 10$ inputs whereas the algorithm in~\cite{Parberry} was estimated to take hundreds of hours on a supercomputer and that in~\cite{Germans} took more than three weeks on a desktop computer.

\section{Sorting Networks}
\label{sec:sorting-networks}

A \defemph{comparator network} $C$ \nop{(see Fig.~\ref{fig:network-example}) }with $n$ channels and depth $d$ is defined as a tuple $C = \tuple{L_1, \ldots, L_d}$ of \defemph{layers} $L_1, \ldots, L_d$. Each layer consists of \defemph{comparators} $\tuple{i, j}$ for pairs of channels $i < j$. Every channel $i$ is required to occur at most once in each layer $L_k$, i.e., $|\{j \setsep \tuple{i,j} \in L_k \vee \tuple{j,i} \in L_k \}| \leq 1$. A layer $L$ is called \defemph{maximal} if no more comparators can be added into $L$, i.e., $|L| = \nhalf$.

An input to a comparator network is a sequence of numbers applied to channels in the first layer. The numbers are propagated through the network; each comparator $\tuple{i, j}$ takes the values from channels $i$ and $j$ and outputs the smaller value on channel $i$ and the larger value on channel $j$. For an input sequence $x_1, \ldots, x_n$ define the value $V(k,i)$ of channel $1 \leq i \leq n$ at layer $k = 0$ (input) to be $V(0,i) = x_i$ and at layer $1 \leq k \leq d$ to be:
\[
V(k, i) = \left\{\begin{array}{ll}
\min(V(k-1, i), V(k-1,j)) \quad & \textrm{if $\tuple{i,j} \in L_k$,} \\
\max(V(k-1, i), V(k-1,j)) & \textrm{if $\tuple{j,i} \in L_k$,} \\
V(k-1,i) & \textrm{otherwise.} \\
\end{array} \right.
\]
The \defemph{output} $C(x)$ of $C$ on $x$ is the sequence $\tuple{V(d, 1), V(d,2), \ldots, V(d, n)}$. See Fig.~\ref{fig:network-example} for an example of a network and its evaluation on an input.

Each comparator permutes the values on two channels and hence the output of a comparator network is always a permutation of the input. A comparator network is called a \defemph{sorting network} if the output $C(x)$ is sorted (ascendingly) for every possible input $x \in \zz^n$. We denote the set of all sorting networks with $n$ channels and depth $d$ by $\bm{S(n,d)}$.

\begin{figure}
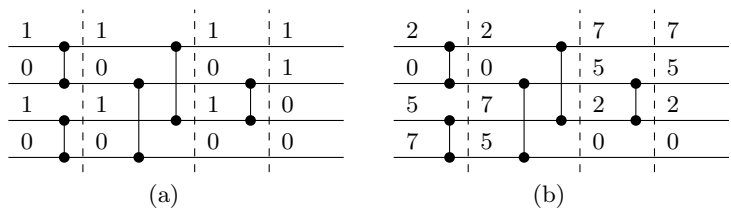

\centering
\begin{tabular}{cc}

\begin{sortingnetwork}{4}{0.7}
\nodelabel{0,1,0,1}
\addcomparator{1}{2}
\addcomparator{3}{4}
\nextlayer
\nodelabel{0,1,0,1}
\addcomparator{1}{3}
\addlayer
\addcomparator{2}{4}
\nextlayer
\nodelabel{0,1,0,1}
\addcomparator{2}{3}
\nextlayer
\nodelabel{0,0,1,1}
\end{sortingnetwork}

&

\begin{sortingnetwork}{4}{0.7}
\nodelabel{7,5,0,2}
\addcomparator{1}{2}
\addcomparator{3}{4}
\nextlayer
\nodelabel{5,7,0,2}
\addcomparator{1}{3}
\addlayer
\addcomparator{2}{4}
\nextlayer
\nodelabel{0,2,5,7}
\addcomparator{2}{3}
\nextlayer
\nodelabel{0,2,5,7}
\end{sortingnetwork}
\\
(a) & (b) 
\end{tabular}
\caption{A comparator network ($L_1 = \{\tuple{1,2},\tuple{3,4}\},L_2=\{\tuple{1,3},\tuple{2,4}\},L_3=\{\tuple{2,3}\}$) with $4$ channels, $5$ comparators, and depth $3$. The channels go from left to right, the first channel is at the bottom, the dashed lines separate the layers. The network on the left is evaluated on the input $\tuple{0,1,0,1}$ and the network on the right on $\tuple{7,5,0,2}$. Diagram shows the values on channels after each layer.}
\label{fig:network-example}
\end{figure}

\nop{A way to measure the efficiency of a sorting network is to count the number of comparators used. However, note that all comparators in a single layer can be evaluated independently in parallel and so the depth of a sorting network is a more appropriate measure of the number of (parallel) computation steps required to sort the input.}

In this work, we are interested in finding the optimal-depth sorting networks for small values of $n$. That is, given $n$, what is the least value of $d$, denoted by $\bm{V(n)}$, such that $S(n,d)$ is nonempty?

Observe that the function $V(n)$ is non-decreasing. Let $C$ be a sorting network with $n$ channels, and construct a network $D$ from $C$ by removing the last channel and all comparators attached to it. Then $D$ is a sorting network with $n-1$ channels: its behaviour on any input is simulated by the first $n-1$ channels of~$C$ if the input to the last channel is set larger than all other inputs ($C(x\infty)$ is $D(x)\infty$, and $C(x\infty)$ is sorted so $D(x)$ is also sorted).

\subsection{Known Bounds on $V(n)$}

Fig.~\ref{fig:previousoptimal} summarises the best bounds on $V(n)$ for $n \leq 16$ channels known before our work. See~\cite{Parberry} for lower bounds on $V(9)$ and $V(10)$, all other numbers appeared already in~\cite{Knuth}. The main contribution of this paper is that $S(11,7)$ and $S(13,8)$ are empty. Thus we improve the lower bounds for $n=11,12$ and $13 \leq n \leq 16$\nop{$n=13,14,15,16$} to $8$ and $9$, respectively, thereby matching the respective upper bounds.

\begin{figure}
\centering
\begin{tabular*}{.85\linewidth}{@{\extracolsep{\fill}}|l|c|c|c|c|c|c|c|c|c|c|c|c|c|c|c|c|}
\hline
$n$ & 1 & 2 & 3 & 4 & 5 & 6 & 7 & 8 & 9 & 10 & 11 & 12 & 13 & 14 & 15 & 16 \\
\hline
Upper bound & 0 & 1 & 3 & 3 & 5 & 5 & 6 & 6 & 7 & 7 & 8 & 8 & 9 & 9 & 9 & 9 \\
\hline
Lower bound & 0 & 1 & 3 & 3 & 5 & 5 & 6 & 6 & 7 & 7 & 7* & 7* & 7* & 7* & 7* & 7* \\
\hline
\end{tabular*}
\caption{Table summarising the best lower and upper bounds known before our work. We improve the starred lower bounds to match the corresponding upper bound.}
\label{fig:previousoptimal}
\end{figure}

One can think of a layer of a comparator network as a matching on $n$ elements: a comparator joins two distinct elements. The number of matchings on $n$ elements grows exponentially in $n$. (See Fig.~\ref{fig:nolayers} for values for $n \leq 13$.) In particular, there are $35696$ matchings on $11$ elements, so to establish the lower bound $V(11) \geq 8$ we have to show that none of the $35696 ^7 \geq 10^{31}$ comparator networks with $11$~channels and depth $7$ is a sorting one. Similarly, to establish $V(13) \geq 9$ we have to consider $568504 ^ 8 \geq 10^{46}$ candidate networks. These numbers immediately make any exhaustive search approach infeasible. In the next section we present techniques to reduce the search space of possible sorting networks, and in Section~\ref{section:sat} we show how to explore this space using a SAT solver.

\section{Search Space Reduction}
\label{sec:reduction}

In the previous section we showed that the number of comparator networks grows very quickly.\nop{We observed in the previous section that there are more than $10^{46}$ candidate networks of $13$ channels and depth 8.} In this section we study general properties of sorting networks with arbitrary numbers of channels and depth, and we show that if $S(n,d)$ is non-empty then it contains a sorting network of a particular form, thus restricting the set of possible candidate networks. For example, for $n=13$ this restricts the set of $568504^2 \geq 3 \cdot 10^{11}$ possible first-two layers to only 212 candidates.

Our arguments build upon and extend those from~\cite{Parberry}, and are based on four technical lemmas given in the following subsections. We make use of the following notation. The set of all layers on $n$ channels is denoted as $\defemph{G}_{\defemph{n}}$. For two networks $C = \tuple{L_1, \ldots, L_p}$ and $D = \tuple{M_1, \ldots, M_q}$ with the same number of channels, the \defemph{composition} $C\fatsemi D$ of $C$ and $D$ is the network $\tuple{L_1, \ldots, L_p, M_1, \ldots, M_q}$. That is, we first apply $C$ and then $D$; for any input $x \in \bb^n$ we have $(C\fatsemi D)(x) = D(C(x))$. A \defemph{prefix} of a network $C$ is a network $P$ such that $C =P\fatsemi Q$ for some network $Q$.  If $L$ is a single layer, we abuse the notation, treat $L$ as a comparator network of depth $1$, and write $L(x)$ for the application of the layer $L$ on input $x$.

\subsection{A Sufficient Sorting Condition}
\label{sec:boolean-sorting}

Before we even start looking for sorting networks it seems necessary to check infinitely many inputs (every $x \in \zz^n$) just to determine whether a comparator network is a sorting one. However, a standard result restricts the set of sufficient inputs to the Boolean ones. Denote $\bb = \{0, 1\}$.

\begin{lemma}[\cite{Knuth}]
\label{lemma:zero-one}
Let $C$ be a comparator network. Then $C$ is a sorting network if and only if $C$ sorts every Boolean input (every $x \in \bb^n$).
\end{lemma}

\subsection{Output-minimal Networks}
\label{sec:output-minimal}

When looking for a sorting network $C = \gennetwork$, we can assume without loss of generality that the first layer $L_1$ is maximal, since by adding comparators to the first layer we can only restrict the set of its possible outputs. We cannot assume that all layers are maximal\nop{ (see~\cite{Knuth} for an example)}, but we can assume that the individual prefixes are maximally sorting in the following sense.
\nop{
\begin{example}
Suppose for example that channels $1$ and $2$ are not used in $L_1$. Let $L_1' = L_1 \cup \tuple{1,2}$ be the layer obtained by adding the comparator $\tuple{1,2}$ to $L_1$ and let $C'$ be the network $C$ with $L_1$ replaced by $L_1'$. Let $y \in \bb^{n-2}$ be any Boolean vector of size $n-2$. Then $L_1'(00y) = L_1(00y), L_1'(10y) = L_1'(01y) = L_1(01y)$ and $L_1'(11y)=L_1(11y)$. Hence, $C'(00y) = C(00y), C'(10y) = C(01y), C'(01y) = C(01y)$ and $C'(11y)=C(11y)$. Since $C$ is a sorting network, these outputs are all sorted, and hence $C'$ is also a sorting network.
\end{example}
}

By $\bm{\netoutput(C)} = \{ C(x) \setsep x \in \bb^n \}$ we denote the set of all possible outputs of a comparator network $C$ on Boolean inputs. The following lemma states that it suffices to consider prefixes $P$ with minimal $\netoutput(P)$.

\begin{lemma}
\label{lemma:subset}
Let $C = P\fatsemi S$ be a sorting network of depth $d$ and $Q$ be a comparator network such that $\depth(P) = \depth(Q)$ and $\netoutput(Q) \subseteq \netoutput(P)$. Then $Q\fatsemi S$ is a sorting network of depth $d$.
\end{lemma}
\begin{proof}
Since $\depth(P) = \depth(Q)$ we have $\depth(Q\fatsemi S) = \depth(P\fatsemi S) = d$.

Let $x \in \bb^n$ be an arbitrary input. Then $Q(x) \in \netoutput(Q) \subseteq \netoutput(P)$. Hence, there is $y\in\bb^n$ such that $Q(x) = P(y)$. Thus, $(Q\fatsemi S)(x) = S(Q(x)) = S(P(y)) = (P\fatsemi S)(y) = C(y)$, which is sorted since $C$ is a sorting network.
\qed
\end{proof}

\subsection{Generalised Sorting Networks and Symmetry}

We further restrict the set of candidate sorting networks by exploiting their symmetry. To facilitate such arguments, we introduce so-called generalised comparator networks~\cite{Knuth} where we lift the condition that the min-channel of a comparator is the one with a smaller index.

Formally, a \defemph{generalised comparator network} $C$ with $n$ channels and depth $d$ is a tuple $C = \tuple{L_1, \ldots, L_d}$ whose layers $L_1, \ldots, L_d$ consists of comparators $\tuple{i, j}$ for channels $i \neq j$, such that each channel occurs at most once in each layer. A comparator $\tuple{i,j}$ is called a \defemph{min-max comparator} if $i < j$ and a \defemph{max-min comparator} otherwise. Channel $i$ receives the minimum and channel $j$ receives the maximum of the values on channels~$i$ and~$j$.

A generalised comparator network can move smaller values to the channel with larger index; we adapt the definition of a sorting network to reflect this. A generalised comparator network $C$ is a \defemph{generalised sorting network} if there exists a permutation $\pi_C$ such that for every $x \in \bb^n$ the value of $C(x)$ is sorted after applying $\pi_C$. That is, if $C(x) = (y_1, \ldots, y_n)$ then we require $(y_{\pi_C(1)}, \ldots, y_{\pi_C(n)})$ to be sorted. It is well known \cite{ParallelComplexityTheory,Knuth} that a generalised sorting network can always be untangled into an ``ordinary'' sorting network of the same dimensions. Furthermore, this operation preserves the ``ordinary'' prefix:

\begin{lemma}[\cite{ParallelComplexityTheory,Knuth}]
\label{lemma:untangle}
If $G$ is a generalised sorting network of depth $d$ then there exist a sorting network $C$ of depth $d$. Furthermore, if $G = P \fatsemi  H$ where $P$ is a comparator network then $C = P \fatsemi  I$ where $I$ is a comparator network.
\end{lemma}

Let $\pi$ be a permutation on $n$ elements. For a comparator $\tuple{i,j}$ we define the comparator $\pi(\tuple{i,j}) = \tuple{\pi(i), \pi(j)}$, and we extend the action of $\pi$ to layers and networks: $\pi(L) = \{\pi(C_1), \ldots, \pi(C_k)\}$ and $\pi(C) = \tuple{\pi(L_1), \ldots, \pi(L_d)}$. Intuitively, applying $\pi$ to a comparator network is equivalent to permuting the channels according to $\pi$; possibly flipping min-max and max-min comparators. Since a generalised sorting network sorts all inputs up to a fixed permutation~($\pi_C)$ of the output, so do its permutations $\pi(C)$ (up to the permutation $\pi_C \circ \pi^{-1}$).

\begin{lemma}[\cite{Parberry}]
\label{lemma:perm}
Let $C$ be a generalised sorting network with $n$ channels and $\pi$ be any permutation on $n$ elements. Then $\pi(C)$ is a generalised sorting network.
\end{lemma}
\nop{
\begin{proof}
Fix $x \in \bb^n$. Since $C$ is a generalised sorting network, there exists a permutation $\pi_C$ such that the vector $\pi_C(C(x))$ is sorted. Hence, $\pi_C(\pi^{-1}(\pi(C)(x)))$ is sorted, so $\pi(C)$ is a generalised sorting network with the corresponding permutation equal to $\pi_C \circ \pi^{-1}$.
\qed
\end{proof}
}

Lemmas~\ref{lemma:zero-one} and \ref{lemma:subset} also hold for generalised comparator networks.

\subsection{First Layer}
\label{sec:first-layer}

We showed in Section~\ref{sec:output-minimal} that if there is a sorting network in $S(n,d)$, then there is one whose first layer is maximal. Now we show that for \emph{any} maximal layer $L$, there exists a sorting network in $S(n,d)$ whose first layer is $L$.

\begin{lemma}[\cite{Parberry}]
\label{lemma:firstfixed}
Let $L$ be a maximal layer on $n$ inputs. If there is a sorting network in $S(n,d)$ there is a sorting network in $S(n,d)$ whose first layer is $L$.
\end{lemma}
\begin{proof}
Let $C = L_1 \fatsemi N$ be a sorting network with $L_1$ its first layer. By Section~\ref{sec:output-minimal}, if $L_1^{+} \supseteq L_1$ is a maximal layer, then $C^{+} = L_1^{+} \fatsemi N$ is also a sorting network. Since $L_1^{+}$ and $L$ are both maximal, there is a permutation $\pi$ such that $\pi(L_1^{+}) = L$. Then, $\pi(C^{+})$ is a generalised sorting network by Lemma~\ref{lemma:perm}. Now, $\pi(C^{+}) = \pi(L_1^{+}) \fatsemi \pi(N) = L \fatsemi \pi(N)$, and by Lemma~\ref{lemma:untangle} there is a comparator network $I$ such that $L \fatsemi I$ is a sorting network and $\depth(L \fatsemi I) = \depth(C^{+}) = \depth(C)$.
\qed
\end{proof}

Lemma~\ref{lemma:firstfixed} allows us to consider only networks with a given maximal first layer. For networks on $n$ inputs we fix the first layer to
\[F_n = \{\tuple{i, \lceil\textstyle\frac{n}{2}\rceil + i} \setsep 1 \leq i \leq \textstyle\nhalf\}.\]

\subsection{Second Layer}
\label{sec:second-layer}

Next we reduce the possibilities for the second layer\footnote{We assume that $n>2$ so that the first layer $F_n$ is not yet a sorting network.}, not to a single candidate but to a small set of candidate second layers. For $n=13$ we arrive at 212 candidates out of the possible $568504$ second layers.

As for the first layer, we can consider second layers modulo permutations of channels. However, we must take into account that the first layer is already fixed to $F_n$, and only consider permutations that leave the first layer intact.

\begin{lemma}[\cite{Parberry}]
\label{lemma:rotation}
Let $\pi$ be a permutation such that $\pi(F_n) = F_n$ and let $L$ be a layer on $n$ channels such that $\pi(L)$ is a layer. If $S(n,d)$ contains a network  with first layer $F_n$ and second layer $L$, it also contains a network with first layer $F_n$ and second layer $\pi(L)$.
\end{lemma}
\nop{
\begin{proof}
For any layer $L$ and a sorting network $C = F_n \fatsemi L \fatsemi N$ of depth $d$, the generalised network $\pi(C) = \pi(F_n) \fatsemi \pi(L) \fatsemi \pi(N) = F_n \fatsemi \pi(L) \fatsemi \pi(N)$ is also sorting and of depth $d$. Hence, by Lemma~\ref{lemma:untangle} there exists a comparator network $I$ such that $F_n \fatsemi \pi(L) \fatsemi I$ is a sorting network of depth $d$.
\qed
\end{proof}
}

Denote by $\bm{H_n}$ the group of permutations on $n$ elements that fix $F_n$. Two layers $L$ and $L'$ are equivalent under $H_n$ if $L' = \pi(L)$ for some $\pi \in H_n$. For any set $S$ of layers, denote by $\bm{R(S)}$ a set of (lexicographically smallest) representatives of $S$ equivalent under $H_n$. Lemma~\ref{lemma:rotation} then implies that it suffices to consider networks with second layers from $R(G_n)$.

Recall from Lemma~\ref{lemma:subset} that it is enough to consider prefixes of comparator networks with minimal sets of possible outputs. We apply a symmetry argument similar to Lemma~\ref{lemma:rotation} to the sets of possible outputs, and observe that it extends to all permutations on $n$ channels. In particular we show that it is enough to consider second layers whose sets of possible outputs are minimal up to any permutation of channels.

\begin{lemma}
\label{lemma:symsubset}
Let $L_a$ and $L_b$ be layers on $n$ channels such that $\netoutput(F_n \fatsemi L_b) \subseteq \pi(\netoutput(F_n \fatsemi L_a))$ for some permutation $\pi$ on $n$ channels. If $S(n,d)$ contains a network with first layer $F_n$ and second layer $L_a$, it also contains a network with first layer $F_n$ and second layer $L_b$.
\end{lemma}
\begin{proof}
Let $C = F_n \fatsemi L_a \fatsemi N$ be a sorting network of depth $d$. Then $\pi(C) = \pi(F_n \fatsemi L_a) \fatsemi \pi(N)$ is a generalised sorting network. Since $\netoutput(F_n \fatsemi L_b) \subseteq \pi(\netoutput(F_n \fatsemi L_a)) = \netoutput(\pi(F_n \fatsemi L_a))$, Lemma~\ref{lemma:subset} implies that $F_n \fatsemi L_b \fatsemi \pi(N)$ is also a generalised sorting network. Then, by Lemma~\ref{lemma:untangle}, there exists a comparator network $I$ such that $F_n \fatsemi L_b \fatsemi I$ is a sorting network of depth $d$.
\qed
\end{proof}

If $\netoutput(F_n \fatsemi L_b) \subseteq \pi(\netoutput(F_n \fatsemi L_a))$ for some permutation $\pi$, we write $L_b \bm{\leq_{po}} L_a$ where $po$ stands for \emph{permuted outputs}. For a set $S$ of layers, denote by $\bm{R_{po}(S)}$ a minimal set of representatives from $S$ such that for each $s\in S$, there is a representative $r \in R_{po}(S)$ such that $r \leq_{po} s$. Lemma~\ref{lemma:symsubset} implies that it suffices to consider second layers from $R_{po}(G_n)$. Fig.~\ref{fig:nolayers} compares numbers of candidate layers $|R_{po}(G_n)|$ and $|R(G_n)|$ with $|G_n|$ for various~$n$.

\subsubsection{Computing the Representatives $R_{po}(G_n)$}

Although we can speed up the search for sorting networks dramatically by only considering second layers from $R_{po}(G_n)$ instead of $G_n$, computing $R_{po}(G_n)$ is a non-trivial task even for $n = 13$.

\nop{Any set $S$ of layers induces a directed graph with an edge from $s$ to $r$ iff $s \geq_{po} r$. The set of representatives $R_{po}(S)$ contains exactly one element from each bottom strongly connected component of this graph.}

Just establishing the inequality $L_a \geq_{po} L_b$ for two layers $L_a$ and $L_b$ involves the comparison of sets $\netoutput(F_n \fatsemi L_b)$ and $\pi(\netoutput(F_n \fatsemi L_a))$, both of size up to $2^n$, for all permutations $\pi$. A naive algorithm comparing all sets of outputs for all pairs of layers thus takes time $O(|G_n|^2 \cdot n! \cdot 2^n)$, and is infeasible\nop{\footnote{$|G_{13}| =  568504, 13! = 6227020800$ and $2^{13} = 8192$}} for $n = 13$. We present three techniques to speed up the computation of $R_{po}(G_n)$.

First we note that in the second layer it is useless to repeat a comparator from the first layer, and in most other cases adding a comparator to the second layer decreases the set of its possible outputs. Call a layer $L$ \defemph{saturated} if it contains no comparator from $F_n$, and its unused channels are either all min-channels, or all max-channels of comparators from $F_n$. Let $S_n$ be the set of all saturated layers on $n$ channels.

\begin{lemma}
Let $n$ be odd and let $L$ be a layer on $n$ channels. There exists a saturated layer $S$ such that $S \leq_{po} L$.
\end{lemma}
\begin{proof}
Let $L$ be any layer on $n$ channels. First construct $L^0$ by removing from $L$ all comparators that also appear in $F_n$. For any input, $L$ and $L^0$ give the same output, so $\netoutput(F_n \fatsemi L) = \netoutput(F_n \fatsemi L^0)$. Next, suppose that $L^0$ is not saturated. Then one of the following holds.
\begin{compactitem}
\item We can add a comparator between a channel $i \leq \nhalf$, which is a min-channel in $F_n$, and a channel $j \geq \lceil\frac{n}{2}\rceil+1$, which is a max-channel in $F_n$ such that $\tuple{i,j}$ is not a comparator from $F_n$. (If $n$ is odd and $L^0$ is not saturated, there are at least 3 unused channels, and we can always choose a pair which is not in $F_n$, not a pair of min-channels, and not a pair of max-channels from $F_n$.) Denote $L^1 = L^0 \cup {\tuple{i,j}}$ and consider the output of $F_n\fatsemi L^1$ on some input $x \in \bb^n$. We will show that $(F_n\fatsemi L^1)(x)$ can also arise as the output of $F_n\fatsemi L^0$.

If $(F_n\fatsemi L^1)(x) \neq (F_n\fatsemi L^0)(x)$, then the output of $(F_n\fatsemi L^0)(x)$ must be 1 on channel $i$ and 0 on channel $j$, and the added comparator $\tuple{i,j}$ flips these values in the output of $F_n\fatsemi L^1$. Since channel $i$ is the min-channel of the comparator $\tuple{i, i+\lceil\frac{n}{2}\rceil}$ in $F_n$, both channels $i$ and $i+\lceil\frac{n}{2}\rceil$ must carry the value 1 in the input $x$. Similarly, since channel $j$ is the max-channel of the comparator $\tuple{j-\lceil\frac{n}{2}\rceil, j}$ of $F_n$, both channels $j$ and $j-\lceil\frac{n}{2}\rceil$ must carry the value 0 in the input $x$. By changing the value of channel $i$ to 0 and the value of channel $j$ to 1, these changes propagate to the output in $F_n\fatsemi L^0$, and yield the same output as that of $F_n\fatsemi L^1$ on $x$. It follows that of $\netoutput(F_n\fatsemi L^1) \subseteq \netoutput(F_n\fatsemi L^0)$.
\item We can add a comparator between some channel $i$ and channel $j = \lceil\frac{n}{2}\rceil$, which is unused in $F_n$, obtaining a layer $L^1$. Similarly as in the previous case we can prove that $\netoutput(F_n\fatsemi L^1) \subseteq \netoutput(F_n\fatsemi L^0)$.
\nop{
\item We can add a comparator between a channel $i \leq \nhalf$, which is a min-channel in $F_n$, and channel $j = \lceil\frac{n}{2}\rceil$, which is unused in $F_n$, yielding a layer $L^1$. We can prove that $\netoutput(F_n\fatsemi L^1) \subseteq \netoutput(F_n\fatsemi L^0)$ similarly as above, except that the channel $j$ is unused in $F_n$ and its input is propagated to the second layer automatically.
\item We can add a comparator between the unused channel and a max-channel in $F_n$, yielding $L^1$, and similarly as above $\netoutput(F_n\fatsemi L^1) \subseteq \netoutput(F_n\fatsemi L^0)$.
}
\end{compactitem}
By induction, we obtain layers $L^1, L^2, \dots$, until some $L^k$ is saturated. Then $\netoutput(F_n \fatsemi L^k) \subseteq \netoutput(F_n \fatsemi L^0) = \netoutput(F_n \fatsemi L)$, so $L^k \leq_{po} L$.
\qed
\end{proof}

Second we note that if two networks are the same up to a permutation $\pi$, then their sets of outputs are also the same up to $\pi$. In particular, $L \leq_{po} \pi(L)$ for any layer $L$ and any $\pi \in H_n$. This observation and the above lemma together imply that it suffices to consider representatives of saturated layers up to permutations from $H_n$ before computing the representatives with respect to $\leq_{po}$.

\begin{lemma}
For odd $n$, we have $R_{po}(G_n) = R_{po}(R(S_n))$.
\end{lemma}

Checking whether a layer is saturated only takes time $O(n^2)$ and computing $R(\cdot)$ involves checking only $\nhalf{}!$ permutations compared to all $n!$ for $R_{po}(\cdot)$. Instead of computing $R_{po}(G_n)$ directly, we first compute $R(S_n)$ and only on this much smaller set we compute the most expensive reduction operation~$R_{po}$. Figure~\ref{fig:nolayers} summarises the number of layers, saturated layers, representatives and representatives modulo rotation for different~$n$.

Finally we show how to compute representatives $R_{po}$. Recall that $L_b \bm{\leq_{po}} L_a$ iff $\netoutput(F_n \fatsemi L_b) \subseteq \pi(\netoutput(F_n \fatsemi L_a))$ for some permutation $\pi$. A necessary condition for $\netoutput(F_n \fatsemi L_b) \subseteq \pi(\netoutput(F_n \fatsemi L_a))$ is that the number of outputs of $(F_n \fatsemi L_a)$ where channel $i$ is set to $1$ is at least the number of outputs of $(F_n \fatsemi L_b)$ where channel $\pi(i)$ is set to $1$. We obtain a similar necessary condition by considering only outputs with value 1 on exactly $k$ channels. For each $i = 1,\dots, n$ and each $k = 0, \dots, n$ we obtain a necessary condition on $\pi$ for $\netoutput(F_n \fatsemi L_b) \subseteq \pi(\netoutput(F_n \fatsemi L_a))$ to hold. These conditions are fast to check and significantly prune the space of possible permutations $\pi$, thereby making the check $L_b \bm{\leq_{po}} L_a$ feasible for any two layers $L_b$ and $L_a$. For $n = 13$ we were able to compute $R(S_n)$ in 2 seconds and subsequently $R_{po}(R(S_n))$ in 32 minutes.

\nop{
We also include in Figure~\ref{fig:nolayers} the numbers of representatives in cases when the first layer $L_1$ is fixed to \nop{a layer other than $F_n$. We show values for the layer $F_n' = \{\tuple{i,\nhalf + i} \setsep 1 \leq i \leq \nhalf \}$ where the unused channel is moved to the end, and} the layer $F'_n = \{\tuple{2i-1,2i} \setsep 1 \leq i \leq \nhalf \}$ used in~\cite{Parberry}. Our choice of $F_n$ allows faster computation of representatives $R_{po}(G_n)$, and yields the smallest $|R_{po}(G_n)|$, of all sampled first layers.
}

\begin{figure}
\begin{center}
\begin{tabular}{|l|@{~}c@{~}|@{~}c@{~}|@{~}c@{~}|@{~}c@{~}|@{~}c@{~}|@{~}c@{~}|@{~}c@{~}|@{~}c@{~}|@{~}c@{~}|@{~}c@{~}|@{~}c@{~}|}
\hline
$n$ & 3  & 4 & 5 & 6 & 7 & 8 & 9 & 10 & 11 & 12 & 13 \\
\hline
\hline
$|G_n|$ & 4 & 10 & 26 & 76 & 232 & 764 & 2620 & 9496 & 35696 & 140152 & 568504 \\
\hline
$|S_n|$ & 2 & 7 & 10 & 51 & 74 & 513 & 700 & 6345 & 8174 & 93255 & 113008 \\
\hline
$|R(G_n)|$ & 4 & 8 & 18 & 28 & 74 & 101 & 295 & 350 & 1134 & 1236 & 4288 \\
\hline
$|R(S_n)|$ & 2 & - & 8 & - & 29 & - & 100 & - & 341 & - & 1155 \\
\hline
$|R_{po}(G_n)|$ & 2 & 2 & 6 & 6 & 14 & 15 & 37 & 27 & 88 & 70 & 212 \\
\hline
\end{tabular}
\end{center}
\caption{
Number of candidates for second layer on $n$ channels. Candidate sets are: $G_n = $ set of all layers, $S_n = $ set of saturated layers, $R(S) = $ set of representatives of $S$ under permutations fixing the first layer, $R_{po}(S) = $ set of representatives of $S$ under permuted outputs. Note that $R(S_n)$ is used to compute $R_{po}(G_n)$ only for odd $n$.}
\label{fig:nolayers}
\end{figure}


\section{Propositional Encoding of Sorting Networks}
\label{section:sat}

In the previous section we showed how to restrict the set of possible first two layers of sorting networks. In this section we describe how to reduce the existence of such a sorting network to the satisfiability of a set of propositional formulas. We then employ the power of modern SAT solvers to determine the satisfiability of the obtained formulas.

Recall that to check whether a comparator network is a sorting one it suffices to consider only its outputs on Boolean inputs (Lemma~\ref{lemma:zero-one}). Now, for Boolean values $x, y \in \bb$ a min-max comparator reduces to: $\min(x, y) = x \wedge y$ and $\max(x, y) = x \vee y$. The authors of~\cite{Germans} observed that a comparator network of a given size can be represented by a propositional formula, and the existence of a sorting network in $S(n,d)$ is equivalent to its satisfiability. We improve upon the work of~\cite{Germans} and give a more natural translation to propositional formulas.

We represent a comparator network with $n$ channels and depth $d$ by Boolean variables $C_n^d = \{g_{i,j}^k\}$ for $1 \leq i < j \leq n$ and $1 \leq k \leq d$. The variable $g_{i,j}^k$ indicates whether the comparator $\tuple{i,j}$ occurs in layer $k$. We then define
\begin{eqnarray*}
\textstyle\once^k_i(C_n^d) & = & \textstyle \bigwedge_{1 \leq i \neq j \neq l \leq n} (\neg g_{\min(i,j),\max(i,j)}^k \vee \neg g_{\min(i,l),\max(i,l)}^k) \qquad\textrm{and}\\
\valid(C) & = & \textstyle\bigwedge_{1 \leq k \leq d,\:1 \leq i \leq n} \textstyle\once^k_i(C_n^d),
\end{eqnarray*}
where $\textstyle\once^k_i(C_n^d)$ enforces that channel $i$ is used at most once in layer $k$, and $\valid(C_n^d)$ enforces that this constraint holds for each channel in every layer, i.e., that $C$ represents a valid comparator network.

Let $x = \tuple{x_1, \ldots, x_n} \in \bb^n$ be a Boolean input and $y = \tuple{y_1, \ldots, y_n}$ be the sequence obtained by sorting $x$. To evaluate the network $C_n^d$ on an input $x$ we introduce variables $v^k_i$ for $0\leq k \leq d$ and $1 \leq i \leq n$ denoting $V(k,i)$--the value of channel $i$ after layer $k$. The correct value of $v^k_i$ is enforced by $\updatechannel^k_i(C_n^d)$ which implements the recursive formula for $V(k,i)$ from Section~\ref{sec:sorting-networks}:
\begin{eqnarray*}
\textstyle\updatechannel^k_i(C_n^d) & = & (\neg \textstyle\used^k_i(C_n^d) \implies (v^k_i \leftrightarrow v^{k-1}_{i})) \wedge \\
                             &    & \textstyle\bigwedge_{1 \leq j < i} \left[ g_{j,i}^k \implies (v^{k}_{i} \leftrightarrow (v^{k-1}_{j} \vee v^{k-1}_{i})) \right] \wedge \\
                             &    &\textstyle\bigwedge_{i < j \leq n} \left[ g_{i,j}^k \implies (v^k_i \leftrightarrow (v^{k-1}_{j} \wedge v^{k-1}_{i})) \right]\qquad\textrm{and} \\
\textstyle\used^k_i(C_n^d) & = & \textstyle\bigvee_{j < i}g_{j,i}^k \vee \bigvee_{i < j}g_{i,j}^k,
\end{eqnarray*}
where the formula $\used^k_i(C_n^d)$ denotes whether channel $i$ is used in layer $k$. We can express the predicate ``$C_n^d(x)$ is sorted'' as:
\begin{eqnarray*}
\sorts(C_n^d,x) & = & \textstyle\bigwedge_{1 \leq i \leq n} (v^0_i \leftrightarrow x_i) \wedge \bigwedge_{\substack{1 \leq k \leq d, \\ 1 \leq i \leq n}} \updatechannel^k_i(C_n^d) \wedge \bigwedge_{1 \leq i \leq n} (v^d_i \leftrightarrow y_i)
\end{eqnarray*}
where the first term ensures that we start with the input $x$, the second term that the $v^k_i$ update appropriately, and the last term that the output is sorted.

\begin{lemma}
\label{lem:sat}
A sorting network with $n$ channels and depth $d$ exists if and only if $\valid(C_n^d) \wedge \bigwedge_{x \in \bb^n} \sorts(C_n^d,x)$ is satisfiable.
\end{lemma}

Further, for inputs of the form $x = 0^py1^q$, we hard-wire the variables $v^k_i$ in the formula $\sorts(C^d_n,x)$ to false for $1 \leq i \leq p$ and to true for $n-q < i \leq n$. These values are implied by the $\updatechannel^k_i(C^d_n)$ formulas (see also Example~\ref{ex:subnet}). However, we find that hard-wiring these values speeds up the SAT solver approximately by a factor of 4 for $n \leq 12$ as the SAT solver is not able to discover them directly by unit propagation.

In Section~\ref{sec:reduction} we showed that it suffices to consider sorting networks with first layer $F_n$ and second layer $S \in R_{po}(G_n)$. We can incorporate such restriction into the propositional formula easily. For the first layer, let $T = \netoutput(F_n)$ be the set of possible outputs, then a sorting network with $n$ channels, depth $d$, and first layer $F_n$ exists if and only if $\valid(C_n^{d-1}) \wedge \bigwedge_{x \in T} \sorts(C_n^{d-1},x)$ is satisfiable. A similar adaptation works for fixing the first two layers; we produce one formula for each $S \in R_{po}(G_n)$ and check the satisfiability of each of them.

Instantiating these SAT formulas and checking their satisfiability was sufficient to establish $V(n)$ for $n \leq 12$ in less than 2 minutes in each case (see Fig.~\ref{table:performance}). A further optimisation substantially reduced the time to establish $V(13)$.

\subsection{Existence of Subnetworks: a Necessary Condition}
\label{section:okienka}

Our final optimisation in showing the nonexistence of sorting network is restricting attention to inputs of the form $0^p y 1^q$. This optimisation is based on the idea that if a comparator network sorts its input, its subnetworks must also sort their respective subinputs. Consider the following example.

\begin{example}
\label{ex:subnet}
Consider the evaluation of a sorting network $C$ on input $0x$ where $x \in \bb^{n-1}$. Since $C$ consists of min-max comparators, the value on the first channel is always~$0$. Hence, also the output of the first channel is~$0$. (See also Fig.~\ref{fig:network-example}.) Let $D$ be the comparator network obtained from $C$ by removing the first channel and all comparators attached to it. Then $C(0x) = 0D(x)$ for all $x$. Requiring that $C(0x)$ is sorted for all $x \in \bb^{n-1}$ is the same as requiring that $D$ is a sorting network. A similar argument can be made for inputs of the form $y1$ for $y \in \bb^{n-1}$, and in general for $0^p y 1^q$ for $y \in \bb^{n-p-q}$.
\end{example}

Let $T^{p,q} = \{t = 0^p x 1^q \setsep t \in T, x \in \bb ^{n-p-q} \} \subseteq T$ be the set of all inputs from $T$ beginning with $p$ zeros and ending with $q$ ones. Intuitively, evaluating a network $C$ on inputs from $T^{p,q}$ exercises only the subnetwork obtained by removing first $p$ and last $q$ channels from $C$.

For subnetwork size $m < n$ let $T_m = \bigcup_{p+q = n-m} T^{p,q}$. Then $T_m \subseteq T$ and so if network $C$ sorts all inputs from $T$ then $C$ sorts all inputs from $T_m$. Therefore, a necessary condition for the existence of a network on $n$ channels and depth $d$ sorting inputs $T$ is the satisfiability of the formula
\[\textstyle\okienka(n,d,m,T) = \valid(C_n^d) \wedge \bigwedge_{x \in T_m} \sorts(C_n^d,x).\]

Empirically, we were always able to find $m$ with $m < n$ such that the resulting formula $\okienka(n,d,m,T)$ was unsatisfiable. Furthermore, the SAT solver established unsatisfiability of this formula significantly faster than for the original formula (see Fig.~\ref{table:performance}).

\section{Experimental Results}
\label{sec:experiments}

In this section we present an experimental evaluation of the described techniques, and show how we used them to obtain bounds on $V(n)$ for $n \leq 16$. We instantiated propositional formulas encoding the existence of a sorting network for various values of $n$ and $d$ and various stages of optimisation as presented in the previous sections.\footnote{Code is available at http://www.cs.ox.ac.uk/people/daniel.bundala/networks/} We checked their satisfiability using an off-the-shelf propositional SAT solver\footnote{MiniSAT version 2.2.0} running on a standard desktop computer\footnote{Linux, CPU: 2.83GHz, Memory: 3.7GiB. All reported times are using a single CPU.}. The times taken by the SAT solver are reported in Fig.~\ref{table:performance}.

\vspace{-0.5em}
\begin{figure}
\begin{center}
{\setlength{\tabcolsep}{1mm}
\begin{tabular}{|@{~}c@{~}||@{~}c@{~}|@{~}c@{~}|@{~}c@{~}|@{~}c@{~}|@{~}c@{~}|@{~}c@{~}|@{~}c@{~}|@{~}c@{~}|@{~}c@{~}|}
\hline
$n$ & 5 & 6 & 7 & 8 & 9 & 10 & 11 & 12 & 13\\
\hline
\hline
$d$ & 4 & 4 & 5 & 5 & 6 & 6 & 7 & 7 & 8\\
\hline
SAT & 0.02s & 0.05s & 1.79s & 1.93s & 864s & 1738s & $>10^5$s & $>10^5$s & - \\
\hline
Fix-1 & 0s & 0s & 0s & 0.02s & 0.5s & 0.5s & 314s & 452s & - \\
\hline
Fix-1 + subnet & 0s & 0s & 0s & 0.01s & 0.27s & 0.26s & 112s & 143s & - \\
\hline
Fix-2 & 0s & 0s & 0.03s & 0.07s & 0.93s & 1.13s &  63s &  87s & 22h23m \\
\hline
Fix-2 + subnet & 0s & 0s & 0.02s & 0.05s & 0.77s & 0.78s & 49s & 48s & 13h1m \\
\hline
\hline
$d$ & 5 & 5 & 6 & 6 & 7 & 7 & 8 & 8 & 9\\
\hline
SAT & 0s & 0.04s & 0.13s & 1.12s & 59.7s & 949s & 1294s & $>10^5$s & - \\
\hline
Fix-1 & 0s & 0s & 0s & 0.01s & 0.20s & 3.6s & 24s & 172s & 1h40m \\
\hline
\end{tabular}
}
\end{center}
\caption{Time required by a SAT solver to solve particular instances of $n$ and $d$ using different variants of propositional formulas: the basic formula from Lemma~\ref{lem:sat} (SAT), a formula fixing the first layer to $F_n$ (Fix-1), formulas fixing the first two layers to $F_n \fatsemi S$ for each $S \in R_{po}(G_n)$ (Fix-2), and the $\okienka(n,d,m)$ versions thereof for appropriate values of $m$ (subnet). The top series corresponds to $d = V(n) - 1$, the largest depth for which no sorting network exists and the formulas are unsatisfiable, the bottom series corresponds to $d = V(n)$ and the formulas are satisfiable. A missing value indicates that the SAT solver ran out of available memory.}
\label{table:performance}
\end{figure}
\vspace{-0.5em}

Our computations confirm the known values of $V(n)$ for $n \leq 10$. Noteworthy is the case $n=9$ where we establish the nonexistence of a sorting network of depth $6$ in less than a second. The specially crafted and low-level optimised program of~\cite{Parberry} was estimated to take 200 hours on the supercomputer Cray-2. Recent work~\cite{Germans} also expressed the existence of such a network as a propositional formula, but their technique by compilation from a higher-level language yields an unnecessarily complicated formula whose SAT checking took over 16 hours.

After 5 minutes of computation when fixing the first layer (2~minutes with the subnetwork optimisation and 1 minute with fixed second layers), we found that $S(11,7)$ is empty. Since $V(11), V(12) \leq 8$ (see Fig.~\ref{fig:previousoptimal}), we have:

\begin{theorem}
The optimal depth of a sorting network with $n=11$ or $12$ channels is eight.
\end{theorem}

Note from Fig.~\ref{table:performance} that checking all Fix-2 formulas for all candidate first-two layers is already faster than checking the single Fix-1 formula; despite the drawback that the SAT solver is restarted for each different second layer. Furthermore, checking the Fix-1 formula requires much more memory, and for the case $n=13$, the SAT solver consumed all available memory (4GB) before finishing. Checking a Fix-2 formula is well within available memory, and different instances for different second layers can be distributed to different computers. This also allows us to start with a small subnetwork size in the subnetwork optimisation and increase it only in instances (second layers) where it yields a satisfiable formula.

For $n = 13$ for each of the $212$ depth-two prefixes $F_{13} \fatsemi L$ we generated a formula $\okienka(13,6,10,T)$ with subnetwork size $m = 10$ and determined that all of them are unsatisfiable in cumulative computation time of 13 hours. Hence, none of the $212$ candidate second layers can be extended to a sorting network.
\begin{theorem}
The optimal depth of a sorting network with $n=13,14,15$ or $16$ channels is nine.
\end{theorem}

Even though we were able to compute lower bounds for $11 \leq n \leq 16$, the case $n=17$ is beyond the scope of current techniques. We leave the depth of the optimal sorting network on $17$ channels as the main open problem of this paper.

\section*{Acknowledgments}
We would like to thank Donald E. Knuth for valuable comments on an earlier draft of this paper which led to strengthening of Lem\-ma~\ref{lemma:symsubset}, reformulation of the $\okienka$ criterion, and inclusion of the hard-wiring optimisation. He also observed that a top-to-bottom reflection of a sorting network is a sorting network, reducing the set of candidate second layers to only $118$ in the case $n=13$.

\bibliographystyle{plain}
\bibliography{bib}

\begin{thebibliography}{1}

\bibitem{AKS}
M.~Ajtai, J.~Koml\'{o}s, and E.~Szemer{\'e}di.
\newblock An 0(n log n) sorting network.
\newblock In {\em Proceedings of the fifteenth annual ACM symposium on Theory
  of computing}, STOC '83, pages 1--9, New York, NY, USA, 1983. ACM.

\bibitem{Batcher}
K.~E. Batcher.
\newblock Sorting networks and their applications.
\newblock In {\em Proceedings of the April 30--May 2, 1968, spring joint
  computer conference}, AFIPS '68 (Spring), pages 307--314, New York, NY, USA,
  1968. ACM.

\bibitem{BoseNelson}
R.~C. Bose and R.~J. Nelson.
\newblock A sorting problem.
\newblock {\em J. ACM}, 9(2):282--296, April 1962.

\bibitem{Knuth}
Donald~E. Knuth.
\newblock {\em The art of computer programming, volume 3: (2nd ed.) sorting and
  searching}.
\newblock Addison Wesley Longman Publishing Co., Inc., Redwood City, CA, USA,
  1998.

\bibitem{Germans}
Andreas Morgenstern and Klaus Schneider.
\newblock Synthesis of parallel sorting networks using sat solvers.
\newblock In Frank Oppenheimer, editor, {\em MBMV}, pages 71--80.
  OFFIS-Institut f{\"u}r Informatik, 2011.

\bibitem{ParallelComplexityTheory}
Ian Parberry.
\newblock {\em Parallel complexity theory.}
\newblock Research notes in theoretical computer science. Pitman, 1987.

\bibitem{Parberry}
Ian Parberry.
\newblock A computer assisted optimal depth lower bound for nine-input sorting
  networks.
\newblock {\em Math. Syst. Theory}, 24:101--116, 1991.

\end{thebibliography}

\end{document}